\documentclass[lettersize,journal]{IEEEtran}
\usepackage[caption=false,font=normalsize,labelfont=sf,textfont=sf]{subfig}
\usepackage{cite}
\usepackage{amsmath,amssymb,amsfonts}
\usepackage{algorithmic}
\usepackage{graphicx}
\usepackage{textcomp}
\usepackage{tikz}
\def\BibTeX{{\rm B\kern-.05em{\sc i\kern-.025em b}\kern-.08em
		T\kern-.1667em\lower.7ex\hbox{E}\kern-.125emX}}
\usepackage{cite}
\usepackage{amsmath,amssymb,amsfonts,mathrsfs}
\usepackage{algorithmic}
\usepackage{graphicx}
\usepackage{textcomp}
\usepackage{times,arydshln}
\usepackage{fancyhdr,graphicx,amsmath,amssymb}
\usepackage[ruled,vlined]{algorithm2e}
\usepackage{enumerate}
\usepackage{blkarray}
\usepackage[english]{babel}
\usepackage{mathtools} 
\usepackage{arydshln}
\usepackage{bbm}
\usepackage{amsfonts}
\usepackage{amssymb}
\usepackage{amsmath}
\usepackage{amsthm}
\usepackage{xcolor}

\renewcommand{\theenumi}{\roman{enumi}}%

\usepackage[ruled,vlined]{algorithm2e}
\usepackage{algorithmic}
\SetKwInput{KwInput}{Input}                % Set the Input
\SetKwInput{KwOutput}{Output}

\def \c {{\bf c}}
\def \p {{\bf p}}
\def \u {{\bf u}}

\def \x {{\bf x}}

\newtheorem{lemma}{{Lemma}}

\newtheorem{construction}{{Construction}}

\newtheorem{remark}{{Remark}}

\newtheorem{proposition}{{Proposition}}

\def \c {{\bf c}}
\def \p {{\bf p}}
\def \u {{\bf u}}

\def \x {{\bf x}}

\begin{document}
	
	\title{On Designing Novel ISI-Reducing Single Error Correcting Codes in an MCvD System}	
	
	\author{Tamoghno Nath,~\IEEEmembership{Student~Member,~IEEE}, Krishna Gopal Benerjee, ~\IEEEmembership{Member,~IEEE,} and  \\ Adrish Banerjee, ~\IEEEmembership{Senior~Member,~IEEE}
		% <-this % stops a space
		%\thanks{This paper was produced by the IEEE Publication Technology Group. They are in Piscataway, NJ.}% <-this % stops a space
		\thanks{The authors are affiliated with the Department of Electrical Engineering, IIT Kanpur, India (email: \{tamoghno, kgopal, adrish\}@iitk.ac.in).}
	}
	
	% \markboth{Journal of \LaTeX\ Class Files,~Vol.~14, No.~8, August~2015}%
	% {Shell \MakeLowercase{\textit{et al.}}: Bare Demo of IEEEtran.cls for IEEE Journals}
	
	\maketitle
	
	% As a general rule, do not put math, special symbols or citations in the abstract or keywords.
	\begin{abstract}
		Intersymbol Interference (ISI) has a detrimental impact on any Molecular Communication via Diffusion (MCvD) system. Also, the receiver noise can severely degrade the MCvD channel performance. However, the channel codes proposed in the literature for the MCvD system have only addressed one of these two challenges independently. In this paper, we have designed single Error Correcting Codes in an MCvD system with channel memory and noise. We have also provided encoding and decoding algorithms for the proposed codes, which are simple to follow despite having a non-linear code construction. Finally, through simulation results, we show that the proposed single ECCs, for given code parameters, perform better than the existing codes in the literature in combating the effect of ISI in the channel and improving the average Bit Error Rate (BER) performance in a noisy channel.
	\end{abstract}
	
	\IEEEpeerreviewmaketitle
	
	\section{Introduction}\label{Sec 1: Intro}
	
	\IEEEPARstart{W}{ithin} the realm of Molecular Communication (MC), MC via Diffusion model is one of the promising communication paradigms due to its low energy consumption and simplicity. It also aligns seamlessly with the requirements of numerous bio-engineering applications \cite{8710366}.
	However, in MCvD systems, ISI, combined with environmental unpredictability, reduces both throughput and data reliability.
	To alleviate the effect of ISI in the MCvD channel, different modulation and detection schemes have been explored in the existing literature \cite{9184816,10100961}.
	In addition to these techniques, channel codes play a pivotal role in enhancing network reliability \cite{7273857,7248461,10149469,10356127,6708566,8972472,10250855,9840783,10207011,10041114}.
	However, most of the Error Correcting Codes, which have been integrated into an MCvD system for improved system performance, use linear ECCs, e.g., Hamming, LDPC, and convolutional codes \cite{7273857,7248461}, and do not consider any ISI-reducing code constraint. In \cite{10149469}, authors have employed a molecular shell mapping (MSM) encoder (along with Reed-Solomon code \cite{10356127}) that places high energy cost symbols at the beginning of codewords, effectively mitigating ISI from the previously transmitted codewords.
	
	Non-linear codes, including ISI-free code and ISI-mtg code \cite{6708566,8972472,10250855 }, have been specifically designed to mitigate  ISI, especially in scenarios involving high-memory channels.
	Shih et al. have introduced ISI-free codes to minimize ISI between intra-codewords while maintaining low decoding complexity \cite{6708566}.
	However, these codes demonstrate an improved BER performance when the data rate is low enough to mitigate the effects of ISI\cite{6708566,8972472}. 
	Also, in \cite{8972472}, an ISI-mtg code with a minimum Hamming distance of one has been suggested, where the codewords follow the Zero Padded (ZP) constraints in the code, which was further refined through
	Huffman coding-based approaches in \cite{10250855}.
	Recent studies have highlighted the impact of bit-1 positioning on ISI reduction. For instance, in \cite{9840783, 10207011}, the authors have constructed a codebook based on the weight and consecutive bit-1 constraints, albeit without inherent error correcting capabilities.
	Therefore, the importance of channel codes that effectively mitigate ISI in noisy channels has motivated the development of the proposed ECCs in this paper. 
	% Consequently, this paper focuses on the algebraic construction of ($n,\mathcal{S},d$) codes that diminish ISI while possessing error correcting capabilities. 
	Consequently, motivated by \cite{10207011}, this paper focuses on the algebraic construction of an ISI-reducing error correcting ($n,\mathcal{S},d$) code, where the average density of bit-1 at $i$-th position in all codewords is non-increasing for $i=1,2,\ldots,n$.
	Also, an ($n,\mathcal{S},d$) code $\mathcal{C}$ is the set of sequences of length $n$, size $\mathcal{S}$ and minimum Hamming distance $d$. The key contributions of this paper are as follows.
	\begin{enumerate}
		\item For positive integers $k$ and $m$ ($m>k$), we have proposed ISI-reducing and single error correcting binary codes $\mathcal{C}_{k,m}$ with the parameter $(k+m+1,2^k,3)$ along with the encoding and decoding mechanisms.
		For an appropriate choice of $k$ and $m$, the asymptotic code rate of the proposed binary code $\mathcal{C}_{k,m}$ approaches $0.5$. 
		\item We have evaluated the performance of the proposed code  $\mathcal{C}_{k,k+1}$ in the presence of noise and ISI: (i) with and (ii) without ``post-encoding operation". We show that the code $\mathcal{C}_{k,k+1}$ exhibits lower ISI than the existing ECCs and this operation further minimizes ISI by distributing consecutive bit-1s, thus improving the BER performance compared to the best-known results in the literature.
	\end{enumerate}
	
	% Organization: 
	The paper structures as follows. Section \ref{Sec 3: Model} introduces an MCvD channel considering ISI, followed by constructing a single Error Correcting Code in Section \ref{Sec 4:Code Design}. Subsequently, Section \ref{Sec 5:Encoding Decoding} discusses encoding and decoding methods, while Section \ref{Sec 6: Performance} evaluates the proposed code's performance followed by the conclusion in Section \ref{Sec 7: Conclusion}.
	
	\section{System Model}\label{Sec 3: Model}
	
	This paper considers a 3-dimensional diffusion model in an unbounded environment, where the point transmitter (Tx) sends a binary sequence $\c = c_1c_2\ldots c_n$ of length $n$ through the MCvD channel using On Off Keying modulation \cite{8972472}, where each binary bit-1 is represented by transmitting $M$ number of molecules in each time interval.
	Finally, a fraction of the transmitted molecules, propagated through the channel by pure diffusion, gets absorbed by the receiver (Rx).
	From \cite{6807659}, the probability of absorbing one molecule by Rx (up to time $t$) is given by
	$F(t) = \frac {r}{r_{0}} \text {erfc}\!\left ({\!\frac {r_{0} - r}{\sqrt {4Dt}}\!}\right)$
	for $r_0>r$, where $r$ is the radius of the Rx, $r_0$ is the distance between Tx and the centre of the Rx and $D$ represents the diffusion coefficient of the molecule.
	For any ($n,\mathcal{S},d$) code $\mathcal{C}$ operating in a channel memory of $L$ ($\geq n$), the expected ISI on the $i$-${\text{th}}$ symbol is 
	\begin{align}
		\mathbb{E}\left[\text{ISI}_i(\mathcal{C})\right] = \frac{1}{\mathcal{S}}\text{ISI}_i{(\mathcal{C})} = \frac{1}{\mathcal{S}}\sum_{\textbf{c}\in\mathcal{C}}\text{ISI}_i{(\c)}.
		\label{Eq: ISI sum}
	\end{align}
	Here, the $i$-th bit ISI in the sequence $\c$ is denoted as ISI$_i{(\c)}=\sum_{k=1}^{i-1}c_{k}p_{i-k+1}$, where $p_{i} = F(it_{s}) - F((i-1)t_{s})$ for $i = 1,2,\ldots,L$ with $t_s$ being the sampling time at the Rx.
	
	\section{Error Correction Code Design}\label{Sec 4:Code Design}
	In this section, we construct a ($k+m+1,2^k,3$) single ECC and explore the ISI performance of the codes.
	In this paper, $\mathbf{1}_{p,q}$ ($\mathbf{0}_{p,q}$) represents an all one (all zero) block with $p$ rows and $q$ columns.
	Additionally, $[.]^T$ signifies the transpose of a matrix  $[.]$. 
	Now, if the $r$-${\text{th}}$ codeword $\c(r)$ of an ($n,\mathcal{S},d$) binary code $\mathcal{C}$ is the sequence $\c(r)=c_1(r)c_2(r)\ldots c_n(r)$, then the average density of bit-1 in the $t$-${\text{th}}$ position for the code $\mathcal{C}$ is $\Tilde{d}(t) = \frac{1}{\mathcal{S}}\sum_{r=1}^{\mathcal{S}}c_t(r)$.
	To construct code $\mathcal{C}_{k,m}$, we define two matrices $\mathbf{U}^{(k)}_{2^k,k}$ of $2^k$ rows and $k$ columns, and $\mathbf{P}_{\binom{m}{i},m}^{(m, i)}$ of $\binom{m}{i}$ rows and $m$ columns satisfying following two properties: 
	\begin{enumerate}
		\item \label{property 1} All rows in  $\mathbf{U}^{(k)}_{2^k,k}$ and $\mathbf{P}_{\binom{m}{i},m}^{(m, i)}$ are arranged in decreasing order based on their respective decimal values.
		\item \label{property 2} All rows in $\mathbf{P}_{\binom{m}{i},m}^{(m, i)}$ have the same weight $i$.
	\end{enumerate}
	The matrices $\mathbf{U}^{(k)}_{2^k,k}$ and $\mathbf{P}_{\binom{m}{i},m}^{(m, i)}$ can be obtained using rec- ursive relations specified in Proposition \ref{prop1} and Proposition \ref{prop2}.
	\begin{proposition}\label{prop1}
		For any positive integer $k$, the matrix $\mathbf{U}^{(k)}_{2^k,k}$ can be obtained recursively by
		\begin{align}\label{msg}
			\mathbf{U}^{(r+1)}_{2^{r+1},r+1} =  
			\left[\begin{array}{cc}
				\mathbf{1}_{2^{r}, 1} & \mathbf{U}^{(r)}_{2^r,r}\\
				\mathbf{0}_{2^{r}, 1} & \mathbf{U}^{(r)}_{2^r,r}
			\end{array}\right]
			\mbox{ for }r=1,2,\ldots,k-1
		\end{align}
		with the initial condition $\mathbf{U}_{2,1}^{(1)} = \left[\begin{array}{cc}
			1 & 0 
		\end{array}\right]^T$.
	\end{proposition}
	\begin{proposition}\label{prop2}
		For any positive integers $m$ and $i$, the matrix $\mathbf{P}_{\binom{m}{i},m}^{(m,i)}$ can be constructed recursively as follows
		\begin{align}
			\mathbf{P}_{\binom{m}{i},m}^{(m,i)} =
			\left[
			\begin{array}{cc}
				\mathbf{1}_{\binom{m-1}{i-1},1} & \mathbf{P}_{\binom{m-1}{i-1},m-1}^{(m-1,i-1)} \\
				\mathbf{0}_{\binom{m-1}{i},1} & \mathbf{P}_{\binom{m-1}{i},m-1}^{(m-1,i)}
			\end{array}
			\right] \mbox{ for } m> i\geq 1
		\end{align}
		with the initial conditions $\mathbf{P}_{\binom{r}{0},r}^{(r,0)} = $ %$\mathbf{P}_{1,r}^{(r,0)}$ = 
		$\left[\mathbf{0}_{1,r}\right]$ for $r = 1,2,\ldots,$ $m-i+1$, and $\mathbf{P}_{\binom{r}{r},r}^{(r,r)} = $ %$\mathbf{P}_{1,r}^{(r,r)}$ = 
		$\left[\mathbf{1}_{1,r}\right]$ for $r=1,2,\ldots,i$. 
	\end{proposition}
	
	Using mathematical induction, one can derive that the matrix $\mathbf{U}^{(k)}_{2^k,k}$ obtained in Proposition \ref{prop1} holds property \eqref{property 1}, and matrix $\mathbf{P}_{\binom{m}{i},m}^{(m,i)}$ obtained in Proposition \ref{prop2} holds both properties \eqref{property 1} and \eqref{property 2}. 
	Also, parameters for $\mathbf{U}^{(k)}_{2^k,k}$ and $\mathbf{P}_{\binom{m}{i},m}^{(m,i)}$ follow from mathematical induction on $k$, $i$ and $m$. Now, the column weight of the matrix $\mathbf{P}_{\binom{m}{i},m}^{(m,i)}$ is discussed in Remark \ref{number of bits 1 and 0}.
	\begin{remark}
		From property \eqref{property 2} and the fact that all sequence of weight $i$ are the rows of the matrix $\mathbf{P}_{\binom{m}{i},m}^{(m,i)}$, the number of rows with bit-$0$ and bit-$1$ at any given index $j\ (1\leq j\leq m)$ in $\mathbf{P}_{\binom{m}{i},m}^{(m,i)}$ are $\binom{m-1}{i}$ and $\binom{m-1}{i-1}$, respectively.
		\label{number of bits 1 and 0}
	\end{remark}
	Now, Construction \ref{Construction 1} outlines the proposed single ECC.
	\begin{construction}\label{Construction 1}
		For any positive integers $k$ and $m$ ($k<m$), consider two matrices $\mathbf{U}^{(k)}_{2^k,k}$ of $2^k$ rows and $k$ columns, and $\mathbf{P}^{(m)}$ $=[
		\begin{array}{cccc}
			\mathbf{P}_{\binom{m}{0},m}^{(m,0)^T} & \mathbf{P}_{\binom{m}{1},m}^{(m,1)^T} & \ldots & \mathbf{P}_{\binom{m}{\tau},m}^{(m,\tau)^T}  
		\end{array}
		]^T$
		of $\sum_{r=0}^{\tau}\binom{m}{r}$ rows and $m$ columns.
		Then, the code is \[\mathcal{C}_{k,m} = \{\c=[\u(r)~\p(r)~\rho(r)]:r=1,2,\ldots,2^k\},\] where $\u(r)$ is the $r$-${\text{th}}$ row of the matrix $\mathbf{U}^{(k)}_{2^k,k}$, $\p(r)$ is the $r$-${\text{th}}$ row of the matrix $\mathbf{P}^{(m)}$, and 
		% $\rho(r)$ is 0 if $i$ is odd in $\mathbf{P}_{\binom{m}{i},m}^{(m,i)}$, otherwise $\rho(r)$ is
		\begin{equation}
			\rho(r)=
			\left\{
			\begin{array}{ll}
				0 & \mbox{ if }i\mbox{ is odd in }\mathbf{P}_{\binom{m}{i},m}^{(m,i)}  \\
				1 & \mbox{ if }i\mbox{ is even in }\mathbf{P}_{\binom{m}{i},m}^{(m,i)}.  
			\end{array}
			\right.
			\label{equ c}
		\end{equation}
	\end{construction}	
	The number of rows in $\mathbf{P}^{(m)}$ is greater than the number of rows in the matrix $\mathbf{U}^{(k)}_{2^k,k}$ for $m> k$. If we consider the weight of the last row of $\mathbf{P}^{(m)}$ to be  $\tau$, then we can determine the value of $\tau$ using the following relationship between $k$ and $\tau$, derived from the property of binomial coefficients:
	\begin{align}\label{k_tau_relation}
		\sum_{r=0}^{\tau-1} \binom{m}{r} < 2^k\leq\sum_{r=0}^{\tau} \binom{m}{r} \mbox{ for } m> k.
	\end{align}
	For instance, if $m=k+1$ then one can use \eqref{k_tau_relation} to determine the maximum weight of any row in the matrix $\mathbf{P}^{(k+1)}$ is $\tau = \left\lceil\frac{k}{2}\right\rceil$.
	
	Now, the parameters of $\mathcal{C}_{k,m}$ are given in Lemma \ref{lem1}.
	\begin{lemma}\label{lem1}
		For any positive integers $k$ and $m$, the length, size and minimum Hamming distance of the code $\mathcal{C}_{k,m}$ are $n=k+m+1$, $\mathcal{S}=2^k$ and $d=3$, respectively.
	\end{lemma}
	\begin{proof}
		The matrices $\mathbf{U}^{(k)}_{2^k,k}$ and $\mathbf{P}^{(m)}$ from Construction \ref{Construction 1} have dimensions $(2^k \times k)$ and $(\sum_{r=0}^{\tau}\binom{m}{r} \times m)$, respectively. The code length follows directly from \eqref{equ c}. By the constraints on $k$ in Construction \ref{Construction 1}, the code size is $\mathcal{S} = \mathrm{min}\{2^k,\sum_{r=0}^{\tau} \binom{m}{r}\} = 2^k$.
		Let $\hat{\mathbf{P}}^{i}=\mathbf{P}_{\binom{m}{i},m}^{(m, i)}$ and $\mathbf{U}^{(k)}=\mathbf{U}^{(k)}_{2^k,k}$.
		Consider $R_{\mathbf{S}}$ as the set of all rows of a matrix $\mathbf{S}$. 
		Any two distinct codewords $\c(r)$ = $[\u(r)~\p(r)~\rho(r)]$ and $\c(s)$ = $[\u(s)~\p(s)~\rho(s)]$ of $\mathcal{C}_{k,m}$ satisfy ($i$) $\u(r),\u(s)\in R_{\mathbf{U}^{(k)}}$, $s.t.$, $\u(r)\neq\u(s)$, ($ii$) $\p(r),\p(s)\in R_{\mathbf{P}^{(m)}}$ $s.t.$ $\p(r)\neq\p(s)$, and ($iii$) $\rho(r)$ and $\rho(s)$ are from \eqref{equ c}. If $\u(r),\u(s)\in R_{\mathbf{U}^{(k)}}$, then $d(\u(r),\u(s))\geq 1$. 
		Now, this leads to two distinct cases. \\
		\textbf{ Case 1} ($\p(r),\p(s)\in R_{\hat{\mathbf{P}}^{i}}$ for some $i$):  
		In this scenario, we have $\sum_{\ell=0}^{i-1} \binom{m}{\ell}<r,s\leq\sum_{\ell=0}^{i} \binom{m}{\ell}\leq2^k$.
		The equal weight of $\p(r)$ and $\p(s)$ leads to $d(\p(r),\p(s))\geq 2$, and, from \eqref{equ c}, $\rho(r)=\rho(s)$. 
		So, the distance $d(\c(r),\c(s))$ = $d(\u(r),\u(s))+d(\p(r),\p(s))+d(\rho(r),\rho(s))\geq3$.\\
		\textbf{ Case 2} (For some distinct $i$ and $j$, $\p(r)\in R_{\hat{\mathbf{P}}^{i}}$ and $\p(s)\in R_{\hat{\mathbf{P}}^{j}}$): 
		In this instance, $r$ holds $\sum_{\ell=0}^{i-1} \binom{m}{\ell}<r\leq\sum_{\ell=0}^{i} \binom{m}{\ell}\leq2^k$, and similarly, $s$ holds $\sum_{\ell=0}^{j-1} \binom{m}{\ell}<s\leq\sum_{\ell=0}^{j} \binom{m}{\ell}\leq2^k$. 
		Also, the weights of $\p(r)$ and $\p(s)$ differ. 
		This leads to ($i$) $d(\p(r),\p(s))\geq2$ and $d(\rho(r),\rho(s))=0$ for even $i-j$, and ($ii$) $d(\p(r),\p(s))\geq1$ and $d(\rho(r),\rho(s))=1$ for odd $i-j$.
		Consequently, $d(\c(r),\c(s))\geq3$ holds for both even and odd cases of $i-j$. 
		This completes the proof. 
	\end{proof}
	\begin{table}[t]
		\begin{center}
			\caption{$(8,8,3)$  code $\mathcal{C}_{3,4}$ with $r=1,2,\ldots,8$.}
			\begin{tabular}{|c|c|p{0.5cm}|c|c|c|p{0.5cm}|c|} \hline
				Index & Message & \multicolumn{2}{c|}{Parity bits} & Index & Message & \multicolumn{2}{c|}{Parity bits} \\ \cline{3-4}\cline{7-8}
				$r$ & bits $\u(r)$  & $\p(r)$ & $\rho(r)$ & $r$ & bits $\u(r)$ &  $\p(r)$ & $\rho(r)$ \\ \hline
				1 & 111 & 0000 & 1 & 5 & 011 & 0001 & 0 \\ \hline   
				2 & 110 & 1000 & 0 & 6 & 010 & 1100 & 1 \\ \hline   
				3 & 101 & 0100 & 0 & 7 & 001 & 1010 & 1 \\ \hline   
				4 & 100 & 0010 & 0 & 8 & 000 & 1001 & 1 \\ \hline   
			\end{tabular}
			\label{mapping}
		\end{center}\vspace{-0.5cm}
	\end{table} 
	An example for the code $\mathcal{C}_{3,4}$ is illustrated in Table \ref{mapping}. In \cite{9840783,10207011}, the authors demonstrate that to minimize the ISI, there is a constraint on the bit-1 locations and the maximum number of consecutive bit-1s for a given sequence of length and weight, resulting in a non-uniformity in the weight distribution.
	So, for any $k$ and $m$, we need to compute the average density of bit-1 at the $t$-th location in the code $\mathcal{C}_{k,m}$ for $t = 1,2,\ldots,k+m+1$ and is given in Lemma \ref{lem2}.
	\begin{lemma}\label{lem2}
		The average density of bit-1 for code $\mathcal{C}_{k,m}$ is
		\begin{enumerate}
			\item $\Tilde{d}(t) = 0.5$ for $1\leq t\leq k$,
			\item $\Tilde{d}(t) \leq \frac{1}{2^k}\sum_{r=1}^{\tau}\binom{m-1}{r-1}$ for $k+1\leq t\leq k+m$, and
			\item $\Tilde{d}(t)\leq\frac{1}{2^k} \sum_{r = 0}^{\left\lfloor\tau/2\right\rfloor}\binom{m-1}{2r}$ for $t = k+m+1$.
		\end{enumerate}
	\end{lemma}
	\begin{proof}
		If $\c$ = $c_1(r)c_2(r)\ldots c_{k+m+1}(r)$ is the $r$-${\text{th}}$ codeword of the binary code $\mathcal{C}_{k,m}$ then there are three cases: \\
		\textbf{Case 1} ($1\leq t\leq k$): From Construction \ref{Construction 1}, density of bit-1 in the matrix $\mathbf{U}^{(k)}_{2^k,k}$ is always 0.5.    \\
		\textbf{Case 2} ($k+1\leq t\leq k+m$): 
		For a given weight $i$ and length $m$ ($m>i\geq0$), each $r$-${\text{th}}$ row of the matrix $\mathbf{P}_{\binom{m}{i},m}^{(m,i)}$, with the $t$-${\text{th}}$ bit as 1, can have $i-1$ additional bit-1s among the remaining $(m-1)$ bits.
		Consequently, in $\mathbf{P}_{\binom{m}{i},m}^{(m,i)}$, rows of weight $i$ at the $t$-${\text{th}}$ position being bit-1 amount to $\binom{m-1}{i-1}$. 
		The maximum possible weight of the $t$-${\text{th}}$ column in $\mathbf{P}^{(m)}$ is bounded by $\sum_{r=0}^{\tau}\binom{m-1}{r-1}$.
		Thus, the result follows for $k+1\leq t\leq k+m$.\\
		\textbf{Case 3 }($t = k+m+1$): 
		From Construction \ref{Construction 1}, the $(k+m+1)$-${\text{th}}$ bit in the $r$-${\text{th}}$ codeword of $\mathcal{C}_{k,m}$ is $\rho(r)=1$ when $i$ is even. Thus, the codewords with bit-1 at $(k+m+1)$-${\text{th}}$ position are $\sum_{r=0}^{\left\lfloor\tau/2\right\rfloor}\binom{m-1}{2r}$.
		Hence, the proof holds for $t = k+ m+ 1$.
	\end{proof}
	From Lemma \ref{lem2}, we can derive the number of codewords with bit-0s and bit-1s at the $t$-th position in the code $\mathcal{C}_{k,m}$ for $1\leq t \leq k+m+1$ and is discussed in the subsequent remark.
	\begin{remark}
		Based on Remark \ref{number of bits 1 and 0} and Construction \ref{Construction 1}, for $k+1\leq t\leq k+m$, the codewords in $\mathcal{C}_{k,m}$ with bit-$1$s at index $t$ are $\Delta_t(1)\leq\sum_{i=1}^\tau\binom{m-1}{i-1}$, while the codewords in $\mathcal{C}_{k,m}$ with bit-$0$s at index $t$ are $\Delta_t(0)\geq\sum_{i=0}^{\tau-1}\binom{m-1}{i}$. 
		For a given index $t$, the difference between the number of bit-0s and bit-1s is
		$$\Delta_t(0) - \Delta_t(1) \geq \left(\sum_{i=0}^{\tau-1}\binom{m-1}{i}\right) - \left(\sum_{i=1}^{\tau}\binom{m-1}{i-1}\right) = 0.$$
		Thus, $\Delta_t(0)\geq\Delta_t(1)$. Additionally, the total number of codewords at index $t$ satisfies $\Delta_t(0) + \Delta_t(1) = 2^k$, implying that $\Delta_t(1) \leq 2^{k-1}$.
		Also, from Property \eqref{property 2} and Construction \ref{Construction 1}, it holds that $\Delta_t(1)\geq\Delta_{t+1}(1)$ for $t=k+1,k+2,\ldots,k+m-1$.
		Note that here $\Delta_t(1) = 2^k \Tilde{d}(t)$ for $1\leq t \leq k+m+1$.
		\label{coll weight diff}
	\end{remark}
	Therefore, from Lemma \ref{lem2} and Remark \ref{coll weight diff}, the column weight $\Delta_t(1)$ of the code $\mathcal{C}_{k,m}$ for $t=1,2,\ldots,k+m+1$ satisfies
	\begin{enumerate}
		\item  $\Delta_t(1)\leq2^{k-1}$ for $t=1,2,\ldots,k+m+1$,
		\item $\Delta_t(1) = 2^{k-1}$ for $t=1,2,\ldots,k$, and
		\item $\Delta_{t+1}(1) \leq \Delta_t(1)$ for $t=k+1,k+2,\ldots,k+m-1$.
	\end{enumerate}
	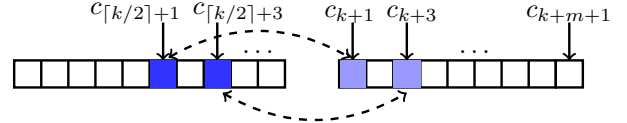
\begin{figure}
		\centering
		\begin{tikzpicture}[line width=0.32mm, scale=0.72] % Adjust scale to fit in two columns
			% Draw first block of length k
			\draw (0,0) rectangle (5,0.5);
			\foreach \x in {0.5,1,1.5,2,2.5,3,3.5,4,4.5} {
				\draw (\x,0) -- (\x,0.5);
			}
			
			% Draw second block of length m+1
			\draw (6,0) rectangle (10,0.5);
			\foreach \x in {6.5,7,7.5,8,8.5,9,9.5} {
				\draw (\x,0) -- (\x,0.5);
			}
			
			% Last small block
			\draw (10,0) rectangle (10.5,0.5);
			
			% Labels with arrows
			\node[above] at (2.25,1) {$c_{\lceil k/2 \rceil+1}$};
			\draw[<-] (2.75,0.5) -- (2.75,1.2);
			
			\node[above] at (4.15,1) {$c_{\lceil k/2 \rceil+3}$};
			\draw[<-] (3.75,0.5) -- (3.75,1.2);
			\node[above] at (4.5,0.45) {$\ldots$};
			
			\node[above] at (6.2,1) {$c_{k+1}$};
			\draw[<-] (6.25,0.5) -- (6.25,1.2);
			
			\node[above] at (7.35,1) {$c_{k+3}$};
			\draw[<-] (7.25,0.5) -- (7.25,1.2);
			
			\node[above] at (8.5,0.45) {$\ldots$};
			
			\node[above] at (10.25,1) {$c_{k+m+1}$};
			\draw[<-] (10.25,0.5) -- (10.25,1.2);
			
			% Highlight swapped blocks
			\fill[blue!80] (2.51,0) rectangle (2.99,0.5);
			\fill[blue!80] (3.52,0) rectangle (3.99,0.5);
			\fill[blue!40] (6.01,0) rectangle (6.49,0.5);
			\fill[blue!40] (7,0) rectangle (7.5,0.5);
			
			% Curved arrows with arrowheads
			\draw[dashed,<->] (2.85,0.55) to[out=30, in=150] (6.2,0.55);
			\draw[dashed,<->] (3.85,-0.1) to[out=-30, in=-150] (7.25,-0.1);
			
			%		 Length indicators
			%		\draw[|-|] (0,-0.7) -- (5,-0.7) node[midway, below] {length = $k$};
			%		\draw[|-|] (6,-0.7) -- (10.5,-0.7) node[midway, below] {length = $m+1$};
		\end{tikzpicture}
		\caption{Post-encoding operation of the code $\mathcal{C}_{k,m}$ for $m> k$.}
		\label{fig_post_encoding}
	\end{figure}
	%		($i$) $\Delta_t(1)\leq2^{k-1}$ for $t=1,2,\ldots,k+m+1$, \\
	%		($ii$) $\Delta_t(1) = 2^{k-1}$ for $t=1,2,\ldots,k$, and \\ 
	%		($iii$) $\Delta_{t+1}(1) \leq \Delta_t(1)$ for $t=k+1,k+2,\ldots,k+m-1$. \\
	For $t=1,3,\ldots,2\lceil\lfloor k/2\rfloor/2\rceil-1$, from \eqref{Eq: ISI sum}, the ISI sum for the code $\mathcal{C}_{k,m}$ at $\left\lceil{k}/{2}\right\rceil+t+1$ and at $k+t+1$ are respectively
	\begin{align}\label{eq_isi_sum_before_swap}
		&	\text{ISI}_{\left\lceil{k}/{2}\right\rceil+t+1}(\mathcal{C}_{k,m}) = 2^{k-1}\sum_{l=2}^{\left\lceil{k}/{2}\right\rceil+t+1}p_{l}, \mbox{ and}\notag\\
		&		\text{ISI}_{k+t+1}(\mathcal{C}_{k,m}) = 2^{k-1}\sum_{l=t+2}^{k+t+1}p_{l}+\sum_{l=2}^{t+1}\Delta_{k+t-l+2}(1)p_{l}.
	\end{align} 
	If we inter-change $\left(\left\lceil{k}/{2}\right\rceil+t\right)$-th and $\left(k+t\right)$-th columns to obtain an equivalent code $\mathcal{C}_{k,m}'$, then from \eqref{Eq: ISI sum}, the ISI sum for the code $\mathcal{C}_{k,m}'$ at $\left\lceil{k}/{2}\right\rceil+t+1$ and $k+t+1$ are respectively
	\begin{align}\label{eq_isi_sum_after_swap}
		&	\text{ISI}_{\left\lceil{k}/{2}\right\rceil+t+1}(\mathcal{C}_{k,m}') = \Delta_{k+t}(1)p_2+2^{k-1}\sum_{l=3}^{\left\lceil{k}/{2}\right\rceil+t+1}p_{l}, \mbox{ and}\notag\\
		& \text{ISI}_{k+t+1}(\mathcal{C}_{k,m}') = \Delta_{k+t}(1)p_{\left\lfloor{k}/{2}\right\rfloor+2}+ 2^{k-1}\sum_{l=\left\lfloor{k}/{2}\right\rfloor+3}^{k+t+1}p_{l}+\notag\\
		&\hspace{1cm}2^{k-1}p_2+ 2^{k-1}\sum_{l=t+2}^{\left\lfloor{k}/{2}\right\rfloor+1}p_{l}+\sum_{l=3}^{t+1}\Delta_{k+t-l+2}(1)p_{l}.
	\end{align} 
	Hence, from \eqref{eq_isi_sum_before_swap} and \eqref{eq_isi_sum_after_swap}, the net gain on ISI for codes $\mathcal{C}_{k,m}'$ and $\mathcal{C}_{k,m}$ is 
	$\left(\mbox{ ISI}_{\left\lceil{k}/{2}\right\rceil+t+1}(\mathcal{C}_{k,m}')+\mbox{ ISI}_{k+t+1}(\mathcal{C}_{k,m}')\right)-\left(\mbox{ ISI}_{\left\lceil{k}/{2}\right\rceil+t+1}(\mathcal{C}_{k,m})+\mbox{ ISI}_{k+t+1}(\mathcal{C}_{k,m})\right)$ = $-(2^{k-1}-\Delta_{\left\lfloor{k}/{2}\right\rfloor+2})\leq0$.
	In fact, for certain values of $t\in\{1,2,\ldots,2\left\lceil \left\lfloor k/2\right\rfloor/2\right\rceil-1\}$, the net gain on ISI is strictly negative.
	%		 This observation motivates the introduction of a ``post-encoding operation" to further minimise ISI in an MCvD channel, as described in Remark \ref{post_encoding} in Section IV.
	This motivates to introduce a  ``post-encoding operation" in Remark \ref{post_encoding}, which distributes consecutive bit-1s across the codeword to further minimize the ISI (Figure \ref{fig_post_encoding}).
	\begin{remark}[Post-Encoding Operation]\label{post_encoding}
		For any positive integers $k$ and $m~(m>k)$, consider any codeword $\c$ = $c_1c_2\ldots c_{k+m+1}$ in the code $\mathcal{C}_{k,m}$. If the binary sequence $\Tilde{\c}=\Tilde{c}_1$ $\Tilde{c}_2\ldots\Tilde{c}_{k+m+1}$ is transmitted over the channel then
		\begin{enumerate}
			\item $\Tilde{c}_{\lceil k/2\rceil + t} = c_{k + t}$ for $t = 1,3,\ldots,2\lceil \lfloor k/2 \rfloor/2\rceil - 1$
			\item $\Tilde{c}_{k + t} = c_ {\lceil k/2\rceil + t}$ for $t = 1,3,\ldots,2\lceil \lfloor k/2 \rfloor/2\rceil - 1$
			\item $\Tilde{c}_t = c_{t}$ for the remaining cases.
		\end{enumerate}
	\end{remark}
	
	For instance, applying the post-encoding operation to the $5$-th codeword in Table \ref{mapping} yields $\Tilde{c}(5) = 01010010$, where the $3$-rd and $4$-th bits of $\mathbf{c}(5)$ are swapped.
	Remark \ref{Remark 3} and Remark \ref{Remark 4} discuss the ISI bounds of code $\mathcal{C}_{k,m}$ for a given code rate.
	\begin{remark}
		From Construction \ref{Construction 1}, for any given rational $\epsilon$ ($0<\epsilon<0.5$), if there exist positive integers $k$ and $\tau$ ($\leq k$) such that ($i$) ${k}/{\epsilon}$ is an integer, ($ii$) $\tau+k+1<{k}/{\epsilon}$, and ($iii$) $\sum_{j=0}^{\tau-1} \binom{\frac{k}{\epsilon}-k-1}{j}<2^k\leq\sum_{j=0}^{\tau} \binom{\frac{k}{\epsilon}-k-1}{j}$, then one can construct the single ECC $\mathcal{C}_{k,k\left({1}/{\epsilon}-1\right)-1}$ with the code rate $\epsilon$ and $\mathrm{ISI}_i(\c)\leq \mathrm{ISI}_{k/\epsilon}([\mathbf{0}_{1,{k}/{\epsilon}-\tau-1}~\mathbf{1}_{1,\tau}~0])$. 
		\label{Remark 3}
	\end{remark} 
	For example, for $\epsilon=1/5$, there exist binary codes $\mathcal{C}_{6,23}$ and $\mathcal{C}_{7,27}$ with $\tau=2$, $s.t.,$ Remark \ref{Remark 3} holds. 
	Also, the expected ISI for the code $\mathcal{C}_{7,27}$ is less than that of $\mathcal{C}_{6,23}$, where the code rates of both codes are the same and $\mathrm{ISI}_i(\c)\leq \mathrm{ISI}_{30}([\mathbf{0}_{1,27}~\mathbf{1}_{1,2}~0])$ for each codeword $\c$ of $\mathcal{C}_{6,23}$ and of $\mathcal{C}_{7,27}$, and $i=1,2,\ldots,30$.
	
	\begin{remark}\label{Remark 4}
		Using Lemma \ref{lem2}, for any given rational $\epsilon$ ($0<\epsilon<0.5$) and given positive integer $\tau$, if there exist $s$ distinct integers $k_1,k_2,\ldots, k_s$, $s.t,$ $k_1<k_2<\ldots<k_s$ and it satisfies all the properties as given in Remark \ref{Remark 3} then the expected ISI for the code $\mathcal{C}_{k_j,k_j(\frac{1}{\epsilon}-1)-1}$ is less than the expected ISI of $\mathcal{C}_{k_i,k_i(\frac{1}{\epsilon}-1)-1}$ for $1\leq i<j\leq s$. Also, the code rates of all the codes $\mathcal{C}_{k_j,k_j(\frac{1}{\epsilon}-1)-1}$ are the same and equal to $\epsilon$.
	\end{remark}	
	From Lemma \ref{lem1}, for any code $\mathcal{C}_{k,m}$, $s.t.,$ $m>k$, the asy- mptotic code rate $\lim\limits_{n\to\infty}\frac{k}{n}$ approaches to $\epsilon< 0.5$. For the par- ticular code $\mathcal{C}_{k,k+1}$,  the asymptotic code rate approaches $0.5$. 

	\section{Encoding and Decoding}\label{Sec 5:Encoding Decoding}
	This section discusses the encoding and decoding of $(k+m+1,2^k,3)$ code $\mathcal{C}_{k,m}$, where the first $k$ positions represent the message bits in encoded bits, and the last $m+1$ positions denote the parity-check bits.
	Any $k$-length message can be encoded using Remark \ref{post_encoding} and Algorithm \ref{Encoding_Algo}. 
	The complexity to obtain the parity-check bits depends on the search operation (step 4 in Algorithm \ref{Encoding_Algo}), which increases with the length and maximum weight ($\tau$) of the codeword.
	However, smaller code parameters, as desired in molecular circuit implementations, keep this complexity low.
	For instance, in the code $\mathcal{C}_{4,5}$, a maximum of \( \binom{5}{\tau} = 10 \) checks are required to compute the parity bits for a given message.
	
	\begin{algorithm}[t]
		\caption{Encoding algorithm for code $\mathcal{C}_{k,m}$}  
		\begin{algorithmic}[1] 
			\REQUIRE Encoded message $\Tilde{c}$
			\ENSURE Message sequence $\u$ = $u_1u_2\ldots u_k$ of length $k$, \\ and length of parity-bit sequence $m+1$
			\STATE msg$_{\mathrm{dec}}$ = $2^k-\sum_{j=1}^{k}u_j2^{k-j}$. %\COMMENT{decimal conversion}.
			\STATE If msg$_{\mathrm{dec}}$ = $0$ then $i=0$ and $\p$ = $\textbf{0}_{1,m}$.
			\STATE If msg$_{\mathrm{dec}}\neq0$ then compute integers $i$ and $r$ by \\ $\sum_{j=0}^{i-1} \binom{m}{j}<\mbox{msg}_{\mathrm{dec}}\leq\sum_{j=0}^{i} \binom{m}{j}$ and \\  $r$ = msg$_{\mathrm{dec}}-\sum_{j=0}^{i-1} \binom{m}{j}$, respectively.
			\STATE $\p$ is the $r$-th row of $\mathbf{P}_{\binom{m}{i},m}^{(m,i)}$ (from Proposition \ref{prop2}).
			\STATE If $i$ is odd then $\rho = 0$, else $\rho = 1$ \COMMENT{($k+m+1$)-th bit}.
			\STATE Return $\c$ =  [$\u~\p~\rho$].
			\STATE Perform post-encoding operation on $\c$ and get $\Tilde{\c}$.
		\end{algorithmic}
		\label{Encoding_Algo} 
	\end{algorithm}	
	
	Upon receiving the post-encoded sequence at the Rx, it first undergoes a ``pre-decoding operation," which is the reverse of the post-encoding operation. Subsequently, the pre-decoded sequence is decoded by Algorithm \ref{Decoding_Algo}, recovering the original message under the following scenarios: ($i$) no errors in the entire $k+m+1$ bits, ($ii$) a single error in the message sequence, allowing recovery from the parity sequence using Algorithm \ref{Decoding_Algo}, or ($iii$) a single error in the parity sequence, resulting in the original decoded message sequence being the first $k$ bits.
	\noindent\textbf{Application:} We can derive the individual expressions for the parity bits by the K-map. For example, for the code $\mathcal{C}_{3,4}$, the output logic expression is $c_4 = \Bar{c}_1\Bar{c}_2+{c}_2\Bar{c}_3, c_5 = {c}_1\Bar{c}_{2}c_{3}+\Bar{c}_1c_{2}\Bar{c}_3, c_6 = \Bar{c}_1\Bar{c}_2c_{3}+c_{1}\Bar{c}_{2}\Bar{c}_3,c_7 = \Bar{c}_1\Bar{c}_2\Bar{c}_3+\Bar{c}_1c_2{c}_3$ from Table \ref{mapping}, and the last parity bit ($c_8$) can be obtained from \eqref{equ c}. 
	One can validate the feasibility of constructing multi-input protein-based logic gates for applications in biological cells \cite{chen2020novo}.
	\begin{algorithm}[t]
		\caption{Decoding algorithm for code $\mathcal{C}_{k,m}$} 
		\begin{algorithmic}[1] 
			\REQUIRE Decoded message $\hat{\textbf{u}}$
			\ENSURE Length of message sequence $k$, and the received sequence $\Tilde{\textbf{v}}$ = $\Tilde{v}_1\Tilde{v}_2\ldots \Tilde{v}_{k+m+1}$ of length $k + m +1$
			\STATE Obtain sequence $\textbf{v}=v_1v_2\ldots v_{k+m+1}$ from $\Tilde{\textbf{v}}$ using ``pre-decoding operation".
			\STATE Extract the message bit sequence $\textbf{u}$ = $v_1v_2\ldots v_{k}$, and the parity bit sequence $[\textbf{p}~\rho] = v_{k+1}v_{k+2}\ldots v_{k+m+1}$.
			\STATE Compute $i = \sum_{j = k+1}^{k+m} v_j$ \COMMENT{weight of $\textbf{p}$}.
			\STATE Based on \eqref{equ c}, obtain $\Hat{\rho}$ from the computed $i$.
			\STATE If $v_{k+m+1}= \Hat{\rho}$ then compute row index $r'$ of the sequence $\textbf{p}$ in $\mathbf{P}_{\binom{m}{i},m}^{(m,i)}$ \COMMENT{$i.e.,$ no error in parity bits}.
			\STATE Get $\mathbf{u}(q)$, i.e., $q\ (=r' + \sum_{j=0}^{i-1}\binom{m}{j})$-th row of $\mathbf{U}^{(k)}_{2^k,k}$.
			\STATE Therefore, the decoded message is $\Hat{\textbf{u}}=\mathbf{u}(q)$.
			\STATE If $v_{k+m+1} \neq \Hat{\rho}$ then $\Hat{\u} = \u$ \COMMENT{$i.e.,$ error in parity bits}.
		\end{algorithmic} 
		\label{Decoding_Algo}
	\end{algorithm}
	\begin{figure*}[h]
		\begin{minipage}[t]{.32\linewidth}
			\centering
			\includegraphics[width=5.5cm]{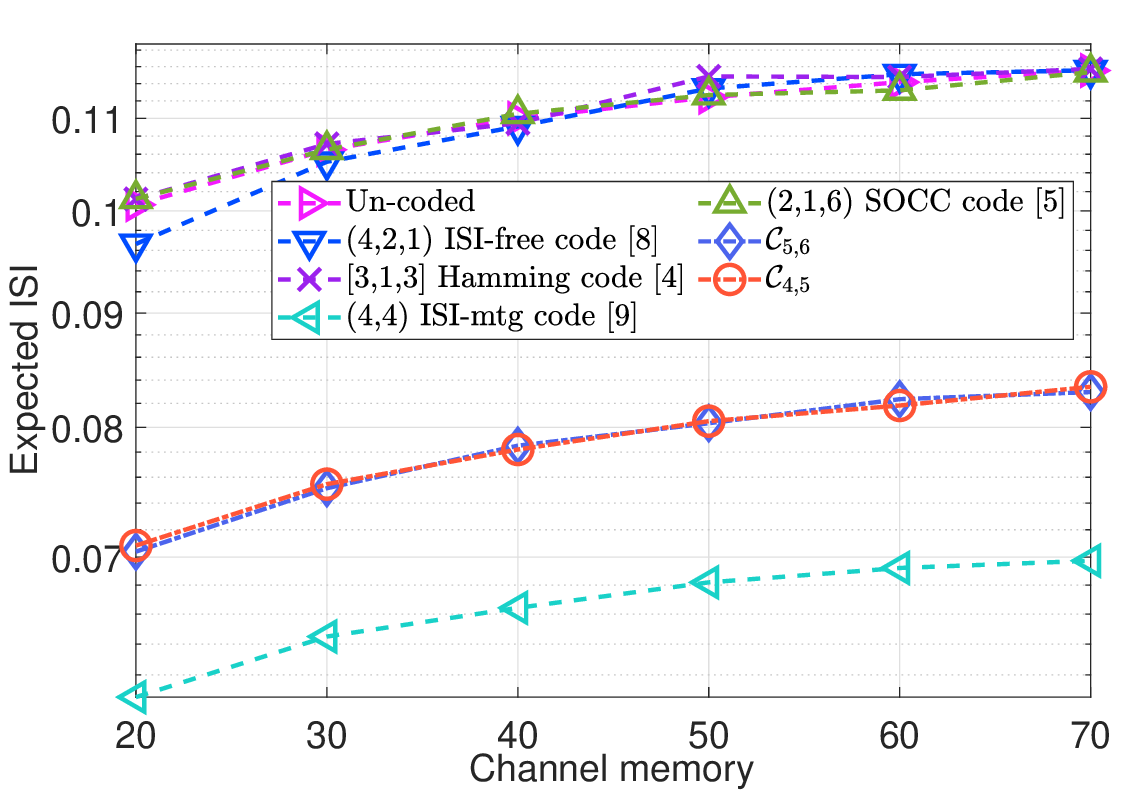}
			\caption{Expected ISI comparison at $t_s = 0.3$s.}
			\label{fig_isi}
		\end{minipage}%
		\begin{minipage}[t]{.33\linewidth}
			\centering
			\includegraphics[width=5.5cm]{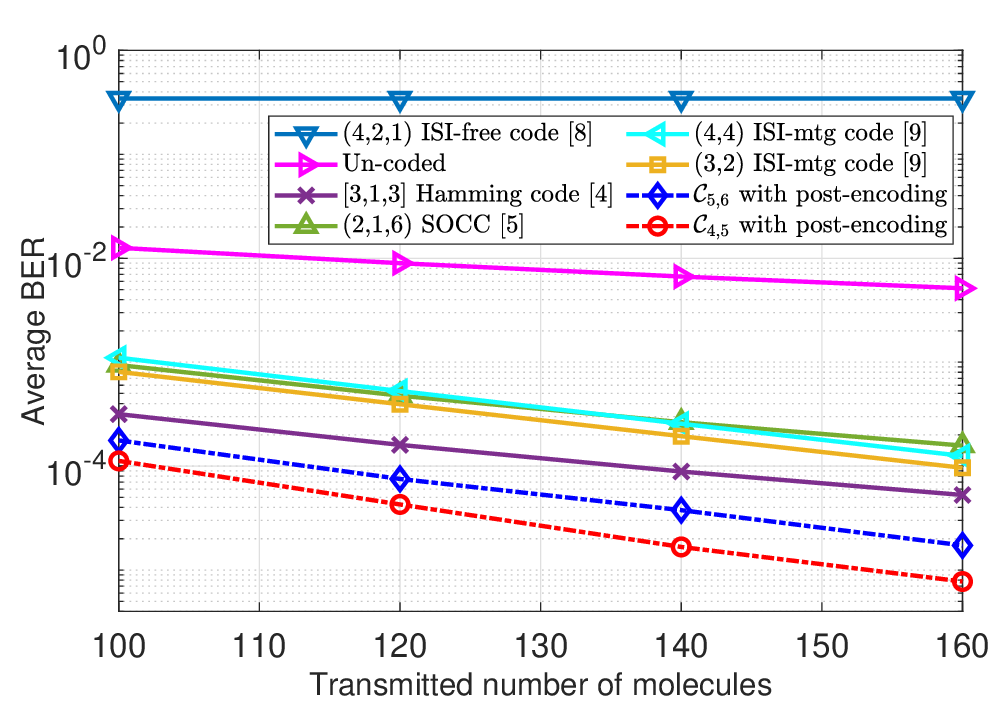}%
			\caption{Average BER with $\sigma_n^2=0$ and $t_s = 0.4$s.}%
			\label{fig_mol}
		\end{minipage}
		\begin{minipage}[t]{.35\linewidth}
			\centering
			\includegraphics[width=5.6cm]{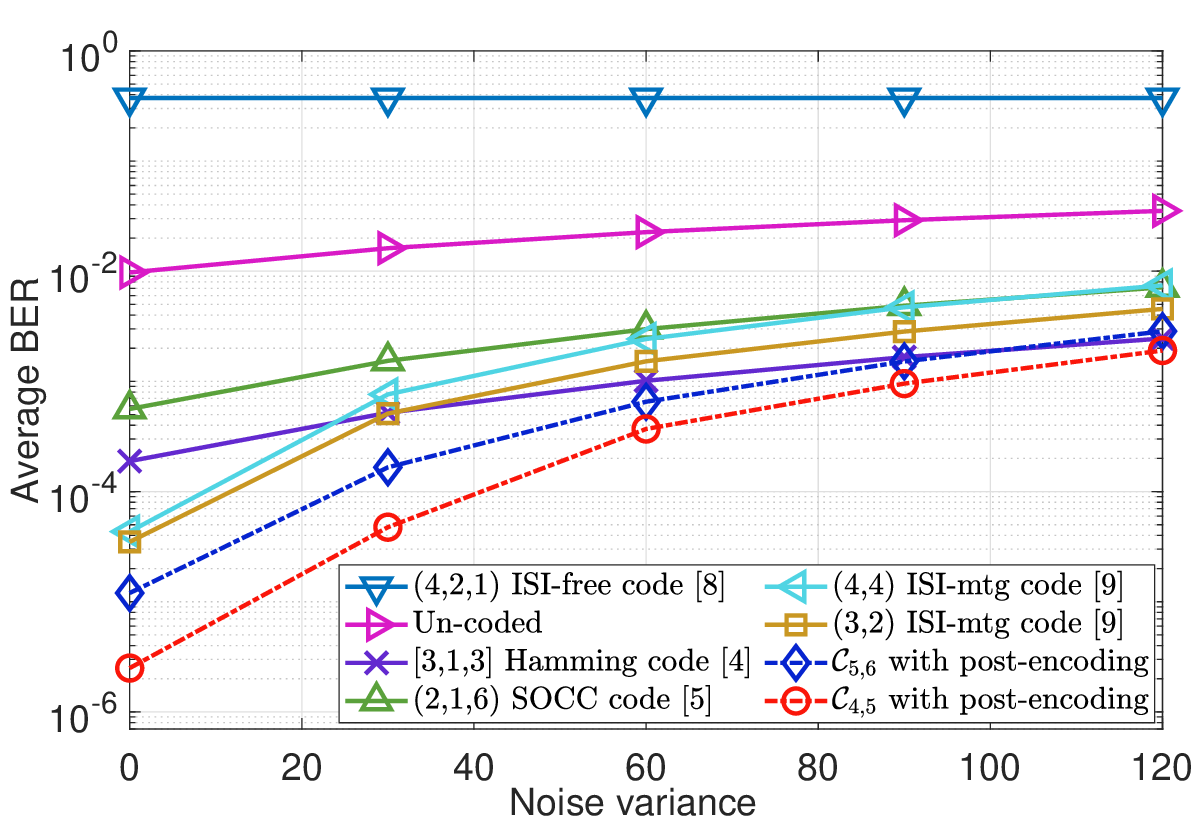}%
			\caption{Average BER with $M = 300$ and $t_s = 0.3$s.}%
			\label{fig_sigma}
		\end{minipage}
	\end{figure*}

	\section{Performance Evaluation}\label{Sec 6: Performance} 
	\begin{figure}[htbp]
		\centering
		\includegraphics[width=7.2cm]{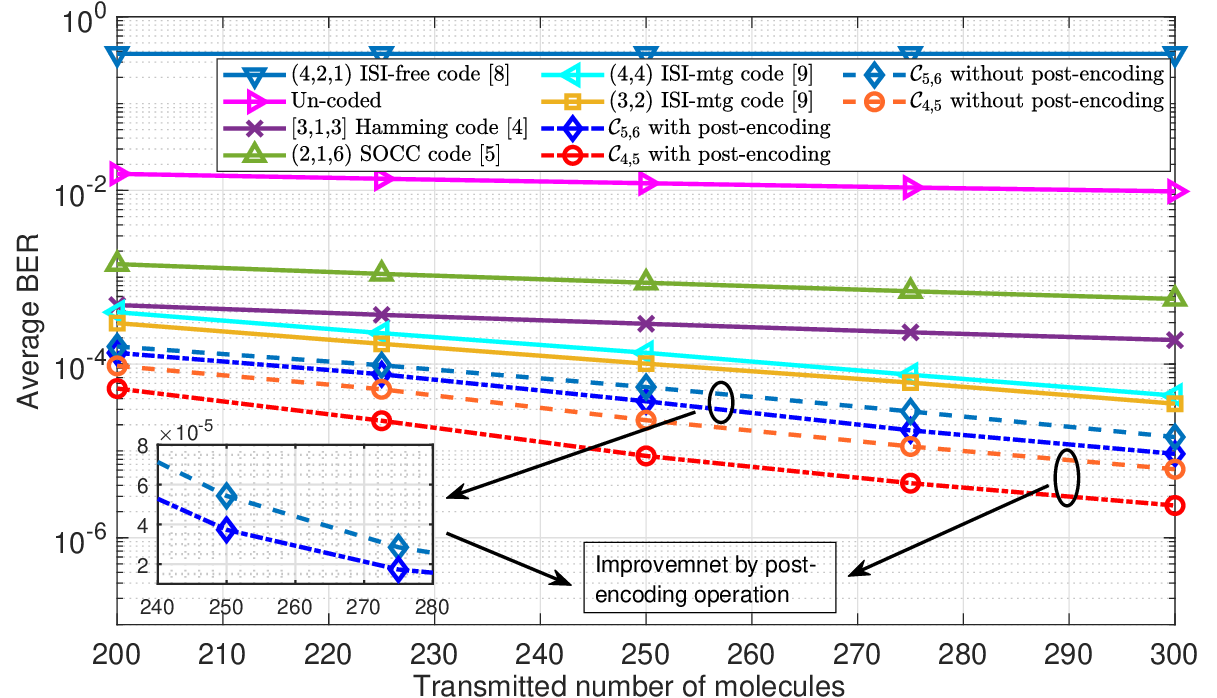}
		\caption{Average BER comparison with $\sigma_n^2 = 0$ and $t_s = 0.3$s.}
		\label{fig_ber_compare_ts_3}
	\end{figure}
	In this section, we evaluate the performance of the proposed ECCs in an MCvD channel with memory. Comparative evaluations include expected ISI and average BER analysis for various linear and non-linear channel codes, such as $[3,1,3]$ cyclic Hamming code \cite{7273857}, $(2,1,6)$ SOCC with six input memory blocks \cite{7248461}, $(4,2,1)$ ISI-free code \cite{6708566}, and $(3,2)$ and $(4,4)$ ISI-mtg codes \cite{8972472}. We compare the BER performance of our proposed $(10,16,3)$ code $\mathcal{C}_{4,5}$ and $(12,32,3)$ code $\mathcal{C}_{5,6}$ against codes with comparable code-rates in the literature.
	Simulation parameters include the values from \cite{8972472}, such that:
	
	(i) $D = 79.4\mu$m$^2$/s, (ii) $r = 5\mu$m, (iii) $r_{0} = 10\mu$m, (iv) $t_s$ = \{0.3s, 0.4s\}, (v) transmitted number of molecules $100 \leq M \leq 300$ and (vi) an additive Gaussian noise with mean $0$ and variance $\sigma_n^{2} ~(\leq 120)$. \\
	\noindent\textbf{Expected ISI Comparison:} In Figure \ref{fig_isi}, we compare the effect of expected ISI on the last bit by transmitting a total of $10^5$ messages over the channel. It is observed that the expected ISI for both the code $\mathcal{C}_{4,5}$ and $\mathcal{C}_{5,6}$ is less than the $[3,1,3]$ Hamming code  \cite{7273857} and $(2,1,6)$ SOCC single ECC \cite{7248461}. 
	Also note that due to a similar bit-1 average density, the ISI performance of $\mathcal{C}_{5,6}$ is nearly identical to that of $\mathcal{C}_{4,5}$.
	In comparison, the $(4,4)$ ISI-mtg code, with a minimum Hamming distance of one, achieves the lowest expected ISI among the considered channel codes for its ZP constraints.  	\\
	\noindent\textbf{BER Performance vs Transmitter:} Figure \ref{fig_mol} and  Figure \ref{fig_ber_compare_ts_3} depict that the code $\mathcal{C}_{4,5}$ and $\mathcal{C}_{5,6}$ at both $t_s = 0.4$s and $0.3$s with $L = 40$, perform better than the existing channel codes in the literature.  
	From \cite{6708566}, the crossover probability in an ISI-free code is proportional to the probability density function of the first hitting time in an MCvD channel without drift. Since, this crossover probability reaches its maximum for $t_s\in [0.2$s, $0.4$s$]$, this results in higher ISI and thus inferior BER performance compared to the uncoded case in Figure \ref{fig_mol} and  Figure \ref{fig_ber_compare_ts_3}.
	From these figures, both codes $\mathcal{C}_{4,5}$ and $\mathcal{C}_{5,6}$, with a higher code rate (code rate of 0.4 and 0.4167, respectively),  achieve a BER improvement over the Hamming and ISI-mitigating codes (code rate of 0.3333).
	%	 which are one of the most effective channel codes for an MCvD channel.
	While, in Figure \ref{fig_mol}, [3,1,3] Hamming code with a larger sampling time ($t_s = 0.4$s), experiences a lesser ISI than at $t_s = 0.3$s, thereby leading to a better BER performance than the ISI-mtg codes. 
	
	Additionally, from Figure \ref{fig_ber_compare_ts_3}, it is evident that without the post-encoding operation, code $\mathcal{C}_{4,5}$ has a BER of $1.125\times 10^{-5}$, while with the operation, it achieves a lower BER of $4.25\times 10^{-6}$ with $M = 275$, confirming the effectiveness of this operation.
	The proposed post-encoding operation constrains the maximum density of bit-1 before a bit-0, leading to a non-linear code construction to mitigate ISI. 
	However, this operation becomes ineffective as the binary linear code, considered in this paper, has an average density of 0.5 for all the bits. 	\\
	\noindent\textbf{BER Performance vs Noise:} As shown in Figure \ref{fig_sigma}, both the proposed codes, $\mathcal{C}_{4,5}$ and $\mathcal{C}_{5,6}$, surpass $[3,1,3]$ Hamming code \cite{7273857}, $(2,1,6)$ SOCC \cite{7248461}, $(4,2,1)$ ISI-free code \cite{6708566}, and $(3,2)$ and $(4,4)$ ISI-mtg codes \cite{8972472} in the channel where the noise variance $(\sigma_n^2)\leq 60$. Notably, the $\mathcal{C}_{4,5}$ code is the best-performing ECC within the considered range of noise ($\sigma_n^2\leq 120$) in an MCvD channel with memory ($L = 40$).
	
	\section{Conclusion}\label{Sec 7: Conclusion}
	This paper introduces a single Error Correcting Code to mitigate ISI in an MCvD channel.
	Comparisons with existing channel codes highlight BER improvements in the presence of noise and ISI, particularly for the code $\mathcal{C}_{k,k+1}$, which exhibit lower ISI levels than conventional linear ECCs.
	Additionally, we introduce the encoding and decoding algorithms for this ECC. The encoded bits undergo a ``post-encoding operation,'' thus reducing consecutive bit-1s in the transmitted sequence and enhancing system reliability.
	
%	\bibliographystyle{IEEEtran} 
%	\bibliography{reference}

\end{document}